\documentclass[10pt]{article}

\usepackage{amsthm,amssymb,amsmath,stmaryrd}
\usepackage{graphics}
\usepackage{bbm}
\usepackage{tikz}
\usepackage[plainpages=false,pdfpagelabels=false,colorlinks=true,citecolor=blue,hypertexnames=false]{hyperref}
\usepackage{color}

\usepackage{dsfont}
\usepackage{epsfig}
\usepackage{mathrsfs}
\usepackage{amsmath}
\usepackage{cases}

\newtheorem{thm}{Theorem}[section]

\newtheorem{ex}[thm]{Example}

\newtheorem{lemma}[thm]{Lemma}
\newtheorem{prop}[thm]{Proposition}
\newtheorem{defn}[thm]{Definition}

\newtheorem{remark}[thm]{Remark}

\makeatletter \@addtoreset{equation}{section}

\newcommand{\bN} { {\mathbb{N}}}
\newcommand{\bC} { {\mathbb{C}}}

\newcommand{\bZ} { {\mathbb{Z}}}

\newcommand{\bL} { {\mathbb{L}}}

\newcommand{\lc}{\operatorname{lc}}
\newcommand{\hc}{\operatorname{hc}}
\newcommand{\hm}{\operatorname{hm}}
\newcommand{\fp}{\operatorname{fp}}
\newcommand{\hp}{\operatorname{hp}}
\newcommand{\pp}{\operatorname{pp}}
\newcommand{\scale}{\operatorname{scale}}
\newcommand{\Li}{\operatorname{Li}}

\def\lc{\operatorname{lc}}

\begin{document}

\title{Additive Decompositions in Primitive Extensions
\thanks{S.\ Chen was supported by the NSFC Grants 11501552, 11688101 and
by the Fund of the Youth Innovation Promotion Association, CAS.}}

\author{  \bigskip
Shaoshi Chen, Hao Du, and Ziming Li \\
$^1$KLMM,\, Academy of Mathematics and Systems Science \\
Chinese Academy of Sciences, Beijing, 100190, China\\
$^2$School of Mathematical Sciences, University of Chinese Academy of Sciences\\ Beijing 100049, (China)\\
\medskip
{\sf schen@amss.ac.cn,\,  duhao@amss.ac.cn, \, zmli@mmrc.iss.ac.cn}
}

\maketitle

\begin{abstract}
This paper extends the classical Ostrogradsky--Hermite reduction for rational functions to more general functions
in primitive extensions of certain types.
For an element $f$ in such an extension $K$, the extended reduction
decomposes $f$ as the sum of a derivative in~$K$ and another element $r$ such that $f$ has an antiderivative in $K$
if and only if $r=0$; and $f$ has an elementary antiderivative over $K$ if and only if $r$ is a linear combination of logarithmic
derivatives over the constants when $K$ is a logarithmic extension. Moreover,~$r$ is minimal in some
sense. Additive decompositions may lead to reduction-based creative-telescoping methods for nested
logarithmic functions, which are not necessarily $D$-finite.
\end{abstract}
\section{Introduction}\label{SECT:intro}
Symbolic integration, together with its discrete counterpart symbolic summation, nowadays has played a crucial role
in building the infrastructure for applying computer algebra tools to solve problems in combinatorics and
mathematical physics~\cite{KKZ2009, KKZ2011, Schneider2013}.
The early history of symbolic integration starts from the first tries of developing programs in LISP to evaluate
integrals in freshman calculus symbolically in the 1960s.
Two representative packages at the time were Slagle's SAINT~\cite{Slagle1961} and Moses's SIN~\cite{Moses1968} which were both
based on integral transformation rules and pattern recognition.
The algebraic approach for symbolic integration is
initialized by Ritt~\cite{Ritt1948} in terms of differential algebra~\cite{Kaplansky1957}, which eventually leads to the Risch algorithm
for the integration of elementary functions~\cite{Risch1969, Risch1970}. The efficiency of the Risch algorithm is further improved
by Rothstein~\cite{Rothsetin1976}, Davenport~\cite{Davenport1981}, Trager~\cite{Trager1984}, Bronstein~\cite{BronsteinThesis, Bronstein1990} etc.
Some standard references on this topic are Bronstein's book~\cite{BronsteinBook} and Raab's survey~\cite{Raab2013} that gives an overview of  the Risch algorithm
and its recent developments.

The central problem in symbolic integration is whether the integral of a given function can be written in
\lq\lq closed form\rq\rq. Its algebraic formulation is given in terms of
differential fields and their extensions~\cite{Kaplansky1957, BronsteinBook}.
A differential field~$F$ is a field together with a derivation $'$ that is an additive map on $F$ satisfying the product rule $(fg)' = f' g + fg'$
for all $f, g\in F$.
A given element~$f$ in $F$
is said to be \emph{integrable} in $F$ if $f=g'$ for some $g\in F$. The problem of deciding whether
a given element is integrable or not in~$F$ is called the \emph{integrability problem} in~$F$.
For example, if $F$ is the field of rational functions, then for $f=1/x^2$ we can
find $g=-1/x$, while for $f=1/x$ no suitable $g$ exists in~$F$. When $f$ is not integrable
in~$F$, there are several other questions we may ask. One possibility is to ask
whether there is a pair~$(g, r)$ in~$F\times F$ such that $f = g' + r$,
where~$r$ is minimal in some sense and $r=0$ if~$f$ is integrable.
This problem is called the \emph{decomposition problem} in~$F$.
Extensive work has been done to solve the integrability and decomposition problems
in differential fields of various kinds.

Abel and Liouville pioneered the early work on the integrability problem
in the 19th century~\cite{Ritt1948}.
In 1833, Liouville provided a first decision procedure for solving the
integrability problem on algebraic functions~\cite{Liouville1833}. For other classes of functions,
complete algorithms for solving the integrability problem are much more recent: 1) the Risch algorithm~\cite{Risch1969, Risch1970} in the case of elementary functions was presented in 1969; 2) the Almkvist--Zeilberger algorithm~\cite{Almkvist1990} (also known as the differential Gosper algorithm) in the case of hyperexponential functions was given in 1990; 3) Abramov and van Hoeij's algorithm~\cite{AbramovHoeij1997} generalized the previous algorithm to the general $D$-finite functions of arbitrary order in 1997.

The decomposition problem was first considered by Ostrogradsky~\cite{Ostrogradsky1845} in 1845 and later by Hermite~\cite{Hermite1872} for rational functions.
The idea of Ostrogradsky and Hermite is crucial for algorithmic treatments of the problem, since it avoids the root-finding of polynomials and only uses
the extended Euclidean algorithm and squarefree factorization to obtain the additive decomposition of a rational function.
This reduction is a basic tool for the integration of rational functions and also
plays an important role in the base case of our work. We will refer this reduction as to
the \emph{rational reduction} in this paper.
The rational reduction has been extended to more general classes of functions including algebraic functions~\cite{Trager1984, Chen2016a}, hyperexponential functions~\cite{GeddesLeLi2004, BCCLX2013},
multivariate rational functions~\cite{BLS2013, Lairez2016}, and more recently including $D$-finite functions~\cite{Chen2018JSC, Hoeven2017}.
Blending reductions with creative telescoping~\cite{Almkvist1990, Zeilberger1991} leads to the fourth and most recent generation of creative telescoping algorithms,
which are called reduction-based algorithms~\cite{BCCL2010, BCCLX2013, BLS2013, Chen2016a, Chen2018JSC}.

The telescoping problem can also be formulated for elementary functions. Two related problems are
how to decide the existence of telescopers for elementary functions and how to compute one if telescopers exist.
Reduction algorithms have been shown to be crucial for solving these two problems. This naturally motivates
us to design reduction algorithms for elementary functions.

In this paper, we extend
the rational reduction to elements in straight and flat towers of primitive extensions (see Definition~\ref{DEF:tower}).
Our extended reductions solve the decomposition problems in such towers without solving any Risch equations (Theorems~\ref{TH:sadd} and~\ref{TH:fadd}), and determine elementary integrability
in such towers when  primitive extensions are logarithmic (Theorem~\ref{TH:elem}).

The remainder of this paper is organized as follows.
We present basic notions and terminologies on differential fields,
and collect some useful facts about integrability in primitive extensions
in Section~\ref{SECT:preliminaries}.  We define the notions of straight and flat towers,
and describe some straightforward reduction processes in Section~\ref{SECT:priext}.
Additive decompositions in straight and flat towers are given in Sections~\ref{SECT:str} and~\ref{SECT:flat},
respectively. The two decompositions are used to determine elementary integrability in Section~\ref{SECT:elem}.
Examples are given in Section~\ref{SECT:ct} to illustrate that the decompositions may be useful to study the telescoping problem for elementary functions that are not $D$-finite.

\section{Preliminaries}\label{SECT:preliminaries}
Let $(F, \, ^\prime)$ be a differential field of characteristic zero, and
let $C_F$ denote the subfield of constants in $F$.
Let $E$ be a differential field extension of $F$.
An element $z$ of $E$ is said to be {\em primitive} over $F$ if~$z^\prime$ belongs to $F$.
If $z$ is primitive and transcendental over $F$ with
$C_{F(z)}=C_F$,
then it is called a {\em primitive monomial} over $F$, which is a special instance of Liouvillian monomials~\cite[Definition 5.1.2]{BronsteinBook}.

Let $z$ be a primitive monomial over $F$ in the rest of this section.
An element $f \in F(z)$  is said to be {\em proper with respect to $z$ or $z$-proper} for brevity
if the degree of its numerator in $z$ is lower than that of its denominator.
In particular, zero is $z$-proper.
It is well-known that $f$ can be uniquely written as the sum of a $z$-proper element and a polynomial in $z$.
They are called the fractional and polynomial parts of~$f$, and denoted by $\fp_z(f)$ and $\pp_z(f)$, respectively.

Let $p$ be a polynomial in $F[z]$.
The degree and leading coefficient of $p$ are denoted by $\deg_z(p)$ and $\lc_z(p)$, respectively.
By~\cite[Thereom 5.1.1]{BronsteinBook}, $p$ is squarefree if and only if $\gcd(p, p^\prime)=1$.
A $z$-proper element is $z$-simple if its denominator is squarefree.
Note that $z$-simple elements are not necessarily $z$-proper
in~\cite{BronsteinBook}, but they are assumed to be  $z$-proper in this paper without loss of generality.

For $S \subset E$, we use $S^\prime$ to denote the set~$\{ f^\prime \mid f \in S\}$. If $S$ is a $C_E$-linear subspace, so is $S^\prime$.
For $m \in \bN$, let
\[  F[z]^{(m)} = \{ p \in F[z] \mid \deg_{z}(p) < m \}. \]
In particular, $F[z]^{(0)} = \{0\}$.

For $f \in F(z)$,  Algorithm {\tt HermiteReduce} in~\cite[page 139]{BronsteinBook} computes a $z$-simple
element $g \in F(z)$ and a polynomial $p \in F[z]$ such that
$
f \equiv g + p \mod F(z)^\prime.
$
This algorithm is an extension of the rational reduction by Ostrogradsky and Hermite.  For rational functions, we have
$p=0$ since all polynomials have polynomial antiderivatives.
Algorithm {\tt HermiteReduce} is fundamental for our approach to additive decompositions in primitive extensions.
\begin{lemma} \label{LM:hp}
Let $g$ be a $z$-simple element in $F(z)$. Then $g=0$ if~$g \in F(z)^\prime + F[z].$
\end{lemma}
\begin{proof} Suppose that~$g \neq 0$. Since $g$ is $z$-proper,
there exists a nontrivial irreducible polynomial $p \in F[z]$ dividing the denominator of~$g$.
Since $g \in F(z)^\prime + F[z],$
there exist $a \in F(z)$ and~$b \in F[z]$ such that $g = a^\prime + b.$
The order of $g$ at $p$ is equal to $-1$.
But the order of $a^\prime$ at $p$ is either nonnegative or less than $-1$ by Theorem 4.4.2~(i) in \cite{BronsteinBook},
and the order of~$b$ at $p$ is nonnegative, a contradiction.
\end{proof}
Every element $f \in F(z)$ is congruent to a unique $z$-simple element $g$ modulo~$ F(z)^\prime + F[z]$
by Algorithm {\tt HermiteReduce} and Lemma~\ref{LM:hp}.
We call $g$ the {\em Hermitian part} of $f$ with respect to $z$, denoted by $\hp_z(f)$.
The map $\hp_z$ is $C_F$-linear on $F(z)$. Its kernel is equal to~$F(z)^\prime + F[z]$.
Thus, two elements have the same Hermitian parts if they are congruent modulo~$F(z)^\prime + F[z]$.
This observation is frequently  used in the sequel.

Now, we collect some basic facts about primitive monomials. They are either straightforward or scattered in~\cite{BronsteinBook}. We list them below for the reader's convenience.
\begin{lemma} \label{LM:lc}
If $p \in F[z]$ and $p \in F(z)^\prime$, then there exists $c \in C_F$ such that
$\lc_z(p) \equiv c z^\prime \mod F^\prime.$
\end{lemma}
\begin{proof}
Assume~$p = r^\prime$ for some $r \in F(z)$.
Then $r \in F[z]$
by Theorem 4.4.2~(i) in \cite{BronsteinBook}.
Set $d = \deg_z(p)$ and $\ell = \lc_z(p)$.
Then $\deg_z(r) \le d+1$ by Lemma 5.1.2 in \cite{BronsteinBook}.
Assume that $r \equiv a z^{d+1} + b z^d \mod F[z]^{(d)}$ for some $a, b \in F$.
Then
$$r^\prime \equiv a^\prime z^{d+1} + ((d+1)az^\prime + b^\prime) z^d \mod F[z]^{(d)}.$$
Since $p=r^\prime$, we have that $a^\prime = 0$ and $\ell = (d+1)az^\prime + b^\prime$.
Hence, $\ell \equiv c z^\prime \mod F^\prime$ with $c = (d+1)a \in C_F$.
\end{proof}
The next lemma will be used to decrease the degree of a polynomial modulo $F(z)^\prime$. Its proof is
a straightforward application of integration by parts.
\begin{lemma} \label{LM:ibp}
We have
$f^\prime z^d \equiv 0 \mod F(z)^{\prime} + F[z]^{(d)}$
for all~$f \in F$ and $d \in \bN$.
\end{lemma}

Recall that an element $f$ in $F$ is said to be a {\em logarithmic derivative in $F$}
if $f=a^\prime/a$ for some nonzero element $a \in F$.
\begin{lemma} \label{LM:ld}
If~$f$ is a $C_F$-linear combination of logarithmic derivatives in~$F(z)$, then
$f = \hp_{z}(f) + r,$
where $r$ is a $C_F$-linear combination of logarithmic derivatives  in~$F$.
\end{lemma}
\begin{proof}
It suffices to assume that $f$ is a logarithmic derivative in~$F(z)$,
because the map $\hp_z$ is $C_F$-linear.

If $f=0$, then we choose $r=0$, which equals $1^\prime/1$.
Otherwise, there exist two monic polynomials $u, v \in F[z]$ and $w \in F$
such that $f = u^\prime/u - v^\prime/v + w^\prime/w$ by the logarithmic derivative identity in~\cite[page 104]{BronsteinBook}.
Note that $u^\prime/u - v^\prime/v$ is $z$-simple by Lemma 5.1.2 in \cite{BronsteinBook} and $w^\prime/w$ is in~$F$. Thus, $\hp_z(f) =  u^\prime/u - v^\prime/v$
and $r = w^\prime/w$.
\end{proof}

\section{Primitive extensions} \label{SECT:priext}
Let $(K_0, ^\prime)$ be a differential field of characteristic zero. Set $C=C_{K_0}$. Consider a tower of differential fields
\begin{equation} \label{EQ:ext}
 K_0 \subset K_1 \subset \cdots \subset K_n,
\end{equation}
where $K_i=K_{i-1}(t_i)$ for all $i$ with $1 \le i \le n$.
The tower given in~\eqref{EQ:ext} is said to be {\em primitive over $K_0$} if
$t_i$ is a primitive monomial over $K_{i-1}$ for all $i$ with $1 \le i \le n$.
The notation introduced in~\eqref{EQ:ext} will be used in the rest of the paper.
\begin{remark} \label{RE:indep}
The derivatives $t_1^\prime, \ldots, t_n^\prime$ are linearly independent over $C$, since $C_{K_n}=C$ in~\eqref{EQ:ext}.
\end{remark}
The following lemma tells us a way to modify the leading coefficient of a polynomial in $K_{n-1}[t_n]$ via integration by parts
and Algorithm {\tt HermiteReduce}.
\begin{lemma} \label{LM:priext}
Let the tower~\eqref{EQ:ext} be primitive with $n  \ge 1$. Then, for all $\ell \in K_{n-1}$ and $d \in \bN$,
there exist  a $t_{n-1}$-simple element $g \in K_{n-1}$ and a polynomial $h \in K_{n-2}[t_{n-1}]$ such that
$$\ell t_n^d \equiv (g+h) t_n^d \mod K_n^\prime + K_{n-1}[t_n]^{(d)}.$$
\end{lemma}
\begin{proof}
By Algorithm {\tt HermiteReduce}, there are $f, g \in K_{n-1}$
with $g$ being $t_{n-1}$-simple, and $h \in K_{n-2}[t_{n-1}]$ such that
$\ell = f^\prime + g + h$. Then $\ell t_n^d = f^\prime t_n^d + (g+h) t_n^d$.
Applying Lemma~\ref{LM:ibp} to the term $f^\prime t_n^d$, we see that the lemma holds.
\end{proof}

Let $\prec$ be the purely
lexicographic ordering on the set of monomials in $t_1, t_2, \ldots, t_n$ with $t_1 \prec t_2 \prec \ldots \prec t_n$.
For all~$i$ with $0 \le i \le n-1$ and  $p \in K_i[t_{i+1}, \ldots, t_n]$ with $p \neq 0$, the head monomial of $p$, denoted by $\hm(p)$, is defined to be
the highest monomial in $t_{i+1}, \ldots, t_n$ appearing in $p$ with respect to $\prec$.
The head coefficient of $p$,
denoted by $\hc(p)$, is defined to be the coefficient of $\hm(p)$, which belongs to $K_i$.
The head coefficient of zero is set to be zero.
The monomial ordering $\prec$ induces a partial ordering on~$K_i[t_{i+1}, \ldots, t_n]$, which is also denoted by $\prec$.
{ \begin{ex}
Let $\xi = t_1 t_2 t_3$. Viewing $\xi$ as an element of $K_0[t_1, t_2, t_3]$, we have $\hm(\xi)=\xi$ and $\hc(\xi)=1$,
while, viewing $\xi$ as an element of~$K_1[t_2, t_3]$, we have $\hm(\xi)=t_2t_3$ and $\hc(\xi)=t_1$,
\end{ex}}

The next lemma will be used in Section~\ref{SECT:flat}. We present it below because it holds for primitive towers.
\begin{lemma} \label{LM:seq}
Let $n \ge 1$.
For  a polynomial $p   \in  K_{n-1}[t_n]$, there are polynomials $p_i   \in  K_i[t_{i+1}, \ldots, t_n]$  such that
$p \equiv \sum_{i=0}^{n-1} p_i \mod K_n^\prime$,  and that $\hc(p_i)$ is $t_i$-simple for all $i$ with $1 \le i \le n-1$.
Moreover, $\deg_{t_n}(p_i) \le \deg_{t_n}(p)$ for all $i$ with $0 \le i\le n-1$.
\end{lemma}
\begin{proof}
 We proceed by induction on $n$.
If $n=1$, then set $p_0 = p$, because there is no requirement on $\hc(p_0)$.
Assume that $n > 1$ and that the lemma holds for $n-1$.

Let $p \in K_{n-1}[t_n]$ and $d = \deg_{t_n}(p)$.
By Lemma~\ref{LM:priext},
$$p \equiv (g+h)t_n^d \mod K_n^\prime + K_{n-1}[t_n]^{(d)},$$
where $g \in K_{n-1}$ is $t_{n-1}$-simple and $h \in K_{n-2}[t_{n-1}]$.
Then there exist $h_j \in K_j[t_{j+1}, \ldots, t_{n-1}]$ such that $h = \sum_{j=0}^{n-2} h_j + u^\prime$ for some $u$ in $K_{n-1}$ and
$\hc(h_j)$ is $t_j$-simple for all $j$ when $1 \le j \le n-2$
by the induction hypothesis.
Furthermore, we set  $h_{n-1} = g$. By Lemma~\ref{LM:ibp},
\begin{equation} \label{EQ:seq}
 p \equiv \sum_{j=0}^{n-1} h_j t_n^d \mod K_n^\prime + K_{n-1}[t_n]^{(d)}.
\end{equation}
We need to argue inductively on $d$.
If $d=0$,  then it is sufficient to set $p_{j}=h_j$ for all $j$ with $0 \le j \le n-1$, as $K_{n-1}[t_n]^{(0)} = \{0\}.$ 
Assume that $d>0$ and that the lemma holds for all polynomials in~$K_{n-1}[t_n]^{(d)}$.
By~\eqref{EQ:seq} and the induction hypothesis on $d$, we have
$$p \equiv \sum_{j=0}^{n-1} h_j t_n^d + \sum_{j=0}^{n-1}\tilde p_j\mod K_n^\prime,$$
where $\tilde p_j$ is in $K_j[t_{j+1}, \ldots, t_n]$, $\hc(\tilde p_j)$ is $t_j$-simple when $j \ge 1$, and $\deg_{t_n}(\tilde p_j) < d$.
Set $p_j = h_j t_n^d + \tilde p_j$. Then $p \equiv \sum_{j=0}^{n-1} p_j \mod K_n^\prime$.
Since $\hc(p_{j})$ is  $\hc(h_j)$ if $h_j \neq 0$ and $\hc(p_j)$ is  $\hc(\tilde p_j)$ if $h_j = 0$, the requirements on each $\hc(p_j)$ with $j \ge 1$ is fulfilled. The induction on $d$ is completed, and so is the induction on $n$.
%
%
\end{proof}
\begin{defn} \label{DEF:tower}
The tower given in~\eqref{EQ:ext} is said to be {\em straight}
if~$\hp_{t_{i-1}}(t_i^\prime) \neq 0$ for all $i$ with $2 \le i \le n$.
The tower is said to be {\em flat} if~$t_i^\prime \in K_0$ for all $i$ with $1 \le i \le n$.
\end{defn}
\begin{ex} \label{EX:tower}
Let $K_0=\bC(x)$ with the usual derivation in $x$.
Let
\[ \text{$\log(x) = \int x^{-1} dx$ \, \, and\, \,  $\Li(x) = \int \log(x)^{-1} dx$}. \]
Then the tower
$$K_0 \subset K_0(\log(x)) \subset K_0(\log(x), \Li(x))$$
is straight, while the tower
$$K_0 \subset K_0(\log(x)) \subset K_0(\log(x), \log(x+1))$$
is flat.
They contain no new constants by Theorem 5.1.1 in~\cite{BronsteinBook}.
\end{ex}
In this paper, we consider additive decompositions for elements in either straight or flat towers
with $K_0=C(t_0)$.
\begin{lemma}  \label{LM:rational}
{Let the tower~\eqref{EQ:ext} be primitive. Assume that~$K_0$ is equal to $C(t_0)$ with the usual derivation in $t_0$.}
Then $\hp_{t_0}(t_1^\prime)$ is nonzero. Moreover,
$\hp_{t_0}(t_i^\prime)$ is nonzero for all $i$ with $2 \le i \le n$ if~\eqref{EQ:ext} is flat.
\end{lemma}
\begin{proof}
By the rational reduction,
 $t_1^\prime {=} u^\prime {+} v$ for some $u, v {\in} K_0$ with $v$ being $t_0$-simple.
 Then $v$ is nonzero, since $C = C_{K_n}$.
 The second assertion can be proved similarly.
\end{proof}

\begin{ex} \label{EX:log}
{  Let $K_{-1}=C$, $K_0=K_{-1}(t_0)$ with the usual derivation in $t_0$, and each $t_i$ be logarithmic in~\eqref{EQ:ext} for all~$i$ with $1 \le i \le n$.
By Lemma~\ref{LM:ld},  $t_i^\prime=\hp_{t_{i-1}}(t_i^\prime) + r_i$, where $r_i \in K_{n-2}$ for all $i$ with $1 \le i \le n$.
The tower is straight if and only if $t_i^\prime \in K_{i-1} \setminus K_{i-2}$
for all~$i$ with $2 \le i \le n$. In addition, $t_1^\prime \in K_0 \setminus K_{-1}$ as $K_0=K_{-1}(t_0)$.}
\end{ex}

\section{Straight towers} \label{SECT:str}

In this section,
we assume that the tower~\eqref{EQ:ext} is straight and  that $K_0 = C(t_0)$
with the usual derivation with respect to~$t_0$. The subfield $C$ of constants is denoted by $K_{-1}$ in recursive definitions and  induction proofs to be carried out.

{  Our idea is reducing a polynomial in $K_{n-1}[t_n]$ to another  of lower degree via integration by parts, whenever it is possible.
The notion of $t_n$-rigid elements describes $r \in K_{n-1}$ such that $r t_{n}^d$ cannot be congruent to
a polynomial of degree lower than~$d$ modulo $K_{n}^\prime$.}
\begin{defn} \label{DEF:rigid}
An element $r \in K_{-1}$ is said to be {\em $t_0$-rigid} if $r=0$.
Let $r \in K_{n-1}$ with
$$f = \fp_{t_{n-1}}(r) \quad \text{and} \quad p = \pp_{t_{n-1}}(r).$$
We say that $r$ is {\em $t_n$-rigid} if $f$ is $t_{n-1}$-simple,
$f \ne c \hp_{t_{n-1}}(t_n^\prime)$ for any nonzero $c \in C$, and $\lc_{t_{n-1}}(p)$ is $t_{n-1}$-rigid.
\end{defn}
Note that zero is $t_n$-rigid, because $\hp_{t_{n-1}}(t_n^\prime)$ is nonzero.
\begin{ex} \label{EX:straight1}
Let $t_0=x$, $t_1=\log(x)$ and $t_2 = \Li(x)$. Let
\[ \ell_1 = \frac{1}{x+k_1} \quad \text{and} \quad \ell_2 = \frac{1}{t_1+k_2} + \ell_1 t_1^2 + x t_1 + x^2. \]
Then $\ell_1$ is $t_1$-rigid if $k_1 \neq 0$ and $\ell_2$ is $t_2$-rigid if $k_1 k_2 \neq 0$.
\end{ex}
The next lemma, together with Lemma~\ref{LM:lc}, reveals that a nonzero polynomial $p$ in $K_{n-1}[t_n]$ with a $t_n$-rigid leading coefficient
has no antiderivative in $K_n$.
\begin{lemma} \label{LM:rigid}
Let $r \in K_{n-1}$ be $t_n$-rigid. If
\begin{equation} \label{EQ:key}
r \equiv c t_n^\prime~\text{mod}~K_{n-1}^\prime
\end{equation}
for some $c \in C$, then both $r$ and $c$ are zero.
\end{lemma}
\begin{proof}
We proceed by induction on $n$.
If $n=0$, then $r = 0$ by Definition~\ref{DEF:rigid}. Thus, $c t_0^\prime \equiv 0 \mod K_{-1}^\prime$.
Consequently,  $c=0$ because $K_{-1}^\prime = \{0\}$ and $t_0^\prime=1$.

Assume that $n>0$ and that the lemma holds for $n-1$.
Set  $f = \fp_{t_{n-1}}(r)$.
Then $f=\hp_{t_{n-1}}(r)$,  since $f$ is $t_{n-1}$-simple by Definition~\ref{DEF:rigid}.
Applying the map $\hp_{t_{n-1}}$ to~\eqref{EQ:key}, we have $f = c \hp_{t_{n-1}}(t_n^\prime)$ by Lemma~\ref{LM:hp}.
Hence, $c=0$ and $f=0$ by Definition~\ref{DEF:rigid}.

Set $p = \pp_{t_{n-1}}(r)$.
Then \eqref{EQ:key} becomes $p \equiv 0 \mod K_{n-1}^\prime$, which, together with Lemma~\ref{LM:lc}, implies that $\lc_{t_{n-1}}(p) \equiv \tilde c t_{n-1}^\prime \mod K_{n-2}^\prime$
for some $\tilde c \in C$. It follows from the induction hypothesis that $\lc_{t_{n-1}}(p)$ is zero, and so is $p$.
Thus, $r$ is zero.
\end{proof}
In $K_{n-1}[t_n]$, we define a class of polynomials  that have no antiderivatives in $K_n$.
\begin{defn} \label{DEF:spoly} {  For $n \ge 0$,
a} polynomial in $K_{n-1}[t_n]$ is said to be {\em $t_n$-straight} if its leading coefficient is $t_n$-rigid.
\end{defn}
\begin{prop} \label{PROP:spoly}
Let $p \in K_{n-1}[t_n]$ be a $t_n$-straight polynomial.
Then $p=0$ if $p \in  K_n^\prime$.
\end{prop}
\begin{proof}
If $n=0$, then $p=0$ by Definition~\ref{DEF:rigid}.
Otherwise, $\lc_{t_n}(p) \equiv c t_n^\prime \mod K_{n-1}^\prime$ for some $c \in C$ by Lemma~\ref{LM:lc}.
Then $\lc_{t_n}(p)=0$ by Lemma~\ref{LM:rigid}. Consequently, $p=0$.
\end{proof}

Next, we reduce a polynomial to a $t_n$-straight one.
\begin{lemma} \label{LM:spoly}
For $p \in K_{n-1}[t_n]$, there exists a $t_n$-straight polynomial $q \in K_{n-1}[t_n]$ with $\deg_{t_n}(q) \le \deg_{t_n}(p)$
such that $p \equiv q \mod K_n^\prime.$
\end{lemma}
\begin{proof}
If $p=0$, then we choose $q=0$. Assume that $p$ is nonzero. We
proceed by induction on $n$.

If $n=0$, then $p \equiv 0 \mod K_0^\prime$, as every element of~$K_{-1}[t_0]$
has an antiderivative in the same ring.

Assume that $n>0$ and that the lemma holds for $n-1$.
Let $p \in K_{n-1}[t_n]$ with degree $d$ and leading coefficient $\ell$.
By Algorithm {\tt HermiteReduce}, there are $f, g, u, v \in K_{n-1}$ with $g, v$ being $t_{n-1}$-simple, and $h, w \in K_{n-2}[t_{n-1}]$
such that
\[ \ell = f^\prime + g + h \quad \text{and} \quad t_n^\prime = u^\prime + v + w. \]

We are going to concoct a new expression for $\ell$ such that
\begin{equation} \label{EQ:trick}
\ell = (c t_n + a)^\prime + r,
\end{equation}
where  $c \in C$, $a \in K_{n-1}$ and $r \in K_{n-1}$ is $t_n$-rigid.  The expression helps us decrease degrees.
To this end,  we consider two cases.

\smallskip \noindent
{\em Case 1.}
Assume $g \neq c v$ for any $c \in C \setminus \{0\}$.
By the induction hypothesis, there exists a $t_{n-1}$-straight polynomial $\tilde h \in  K_{n-2}[t_{n-1}]$
such that $h= b^\prime + \tilde h$ for some $b \in K_{n-1}$.
Then $\ell =  a^\prime + r$, where $a = f + b$ and $r = g + \tilde h$.

\smallskip \noindent
{\em Case 2.}
Assume $g = c v$ for some $c \in C \setminus \{0\}$.
By the induction hypothesis, there exists a $t_{n-1}$-straight polynomial $\tilde h \in  K_{n-2}[t_{n-1}]$
such that $h-c w = b^\prime + \tilde h$ for some $b \in K_{n-1}$.
Then $\ell =  (c t_n + a)^\prime + r$, where $a = f - c u + b$ and $r = \tilde h$.

In both cases, $r$ is $t_n$-rigid by Definition~\ref{DEF:rigid}.

If $d=0$, then $p = \ell$.  By~\eqref{EQ:trick}, we have  $p \equiv r \mod K_n^\prime$.
Let $q = r$, which is $t_n$-straight by Definition~\ref{DEF:spoly}.

Assume that $d >0$ and each polynomial in~$K_{n-1}[t_n]^{(d)}$ is congruent to a $t_n$-straight polynomial modulo~$K_n^\prime$.
By~\eqref{EQ:trick},  Lemma~\ref{LM:ibp} and the equality $c t_n^\prime t_n^d {=} \left(\frac{c}{d+1} t_n^{d+1}\right)^\prime$,
we have
$$p \equiv r t_n^d  + \tilde q \mod K_n^\prime$$
for some $\tilde q \in K_{n-1}[t_n]^{(d)}$. If $r \neq 0$, then set $q =  r t_n^d  + \tilde q.$
Otherwise, applying the  induction hypothesis on $d$ to $\tilde q$ yields a $t_n$-straight polynomial $q$
with $p \equiv q \mod K_n^\prime$.
\end{proof}
\begin{ex} \label{EX:straight2}
Consider the integral
\[ \int \log(x) \Li(x)^2 \, dx. \]
With the notation introduced in Example~\ref{EX:straight1}, we reduce the integrand $t_1t_2^2$.
We have that $\lc_{t_2}(t_1t_2^2)=t_1$.
Since $t_1$ is not $t_2$-rigid,
$t_1t_2^2$ can be reduced.
In fact, $t_1 t_2^2 = x^\prime t_1 t_2^2$. By Lemma~\ref{LM:ibp} and a straightforward calculation, we get
$$t_1 t_2^2 = \left( x t_1 t_2^2 - x t_2^2 - x^2 t_2 \right)^\prime + \frac{2x}{t_1} t_2 + \frac{x^2}{t_1}.$$
Since $2x/t_1$ is $t_2$-rigid, we have that $(2x/t_1)t_2 + (x^2/t_1)$ is $t_2$-straight.
Hence, $t_1 t_2^2$ has no antiderivative in $C(x, t_1, t_2)$
by Proposition~\ref{PROP:spoly}.
\end{ex}

Below is an additive decomposition in a straight tower.
\begin{thm} \label{TH:sadd}
For $f \in K_n$, the following assertions hold.
 \begin{enumerate}
 \item[(i)] There exist a $t_n$-simple element $g \in K_n$ and a $t_n$-straight polynomial $p\in K_{n-1}[t_n]$ such that
 \begin{equation} \label{EQ:sadd}
   f \equiv g + p \mod K_n^\prime.
 \end{equation}
 \item[(ii)] $f \in  K_n^\prime$ if and only if both $g$ and $p$ in~\eqref{EQ:sadd} are zero.
 \item[(iii)] If  $f \equiv \tilde g + \tilde p \mod K_n^\prime$, where $\tilde g \in K_n$ is a $t_n$-simple element
 and $\tilde p \in K_{n-1}[t_n]$, then $g = \tilde g$ and $\deg_{t_n}(p) \le \deg_{t_n}(\tilde p).$
 \end{enumerate}
\end{thm}
\begin{proof}
 (i) By Algorithm {\tt HermiteReduce},
 there exist a $t_n$-simple {  element} $g \in K_n$ and a polynomial $h\in K_{n-1}[t_n]$ such that
$$f \equiv g + h \mod K_n^\prime.$$
 By Lemma~\ref{LM:spoly}, $h$ can be replaced by a $t_n$-straight polynomial $p$.

 (ii) Since $f \in K_n^\prime$,  the congruence \eqref{EQ:sadd} becomes $g + p \equiv 0 \mod K_n^\prime$.
 Applying the map $\hp_{t_n}$ to the new congruence, we have $g=0$, because $g = \hp_{t_n}(g+p)$.
 Thus, $p=0$ by Proposition~\ref{PROP:spoly}.

 (iii) Since $g - \tilde g \equiv  \tilde p - p \mod K_n^\prime$,  we have $g = \tilde g$
 by Lemma~\ref{LM:hp}. If $\deg_{t_n}(\tilde p) < \deg_{t_n}(p)$, then $p-\tilde p$ is $t_{n}$-straight,
 because $\lc_{t_n}(p-\tilde p)$ equals  $\lc_{t_n}(p)$.
 So $p-\tilde p=0$ by Proposition~\ref{PROP:spoly}, a contradiction.
\end{proof}
\begin{ex} \label{EX:straight3}
Consider the integral
$$\int \frac{1}{\Li(x)^2} + \log(x) \Li(x)^2 \, dx.$$
The integrand is $f := 1/t_2^2 + t_1 t_2^2$, in which the notation is introduced in Example~\ref{EX:straight1}.
By Algorithm {\tt HermiteReduce}, we have
 $$f = (-t_1/t_2)^\prime + x/t_2 + t_1 t_2^2.$$ By Theorem~\ref{TH:sadd} and
Example~\ref{EX:straight2}, $f$ has no antiderivative in~$C(x, t_1, t_2)$.
\end{ex}

 \section{Flat towers} \label{SECT:flat}
 In this section, we let the tower~\eqref{EQ:ext} be flat. The ground field $K_0$ will be specialized to $C(t_0)$ later in this section.
 We are not able to fully carry out the same idea in Section~\ref{SECT:str},
 because $\hp_{t_{i-1}}(t_i^\prime) = 0$ for all $i=2, \ldots, n$. This spoils Lemma~\ref{LM:rigid} and Proposition~\ref{PROP:spoly}.
 So we need to study integrability in a flat tower differently.

  This section is divided into two parts. First, we extend Lemma~\ref{LM:ibp}  to
 the differential ring $K_0[t_1, \ldots, t_n]$. Second, we present a flat counterpart of the results in Section~\ref{SECT:str}.

 \subsection{Scales}
Let us denote $K_0[t_1, \ldots, t_n]$ by $R_n$.
For a monomial $\xi$ in $t_1$, \ldots, $t_n$,
the $C$-linear subspace $\{ p \in R_n \mid p \prec \xi \}$
is denoted by $R_n^{(\xi)}$. The notion of scales is motivated by
the following example.
\begin{ex} \label{EX:scale1}
Let $n=2$, and $\xi_0=1, \xi_1=t_1$ and $\xi_2 = t_2$.
And let $\ell = t_1^\prime + t_2^\prime$.
Using integration by parts, we find three congruences
$$\ell \xi_0 \equiv 0~\operatorname{mod}~K_2^\prime, \quad  \ell \xi_1 \equiv -t_1^\prime t_2 ~\operatorname{mod}~K_2^\prime, \quad
\,\, \ell \xi_2 \equiv - t_2^\prime t_1~\operatorname{mod}~K_2^\prime.$$
The first and third congruences lead to monomials  lower than~$\xi_0$ and~$\xi_2$, respectively.
But the second one leads to $t_2$, which is higher than~$\xi_1$. The notion of scales aims to prevent
the second congruence from the reduction to be carried out.
\end{ex}

\begin{defn} \label{DEF:scale}
Let $p \in R_n \setminus \{0\}$ and
$ \hm(p) = t_1^{e_1}  \cdots t_n^{e_n}.$ The {\em scale} of $p$ with respect to $n$
is defined to be $s$ if $e_1=0$, \ldots,  $e_{s-1} = 0$ and $e_s > 0$.
Let $p \in K_0$. The scale of $p$ with respect to $n$ is defined to be~$n$.
The scale of $p$ with respect to $n$ is denoted by $\scale_n(p)$.
\end{defn}
\begin{ex}
Let $\xi_0=1$, $\xi_1=t_1 t_2$ and $\xi_2=t_3^2$. Regarding $\xi_0, \xi_1$ and $\xi_2$ as elements in $K_0[t_1, t_2, t_3]$, we have that
$\scale_3(\xi_0)=3,$ $\scale_3(\xi_1)=1$ and $\scale_3(\xi_2)=3$; while, regarding them as elements in $K_0[t_1, t_2, t_3, t_4]$,
we have that $\scale_4(\xi_0)=4,$ $\scale_4(\xi_1)=1$ and $\scale_4(\xi_2)=3$.
\end{ex}
Notably, if $p \in K_0$, then the scale of $p$ with respect to $n$ is equal to $n$, which varies as $n$ does. Otherwise, the scale is fixed by $\hm(p)$ no matter
in which ring $p$ lives.

The next lemma extends Lemma~\ref{LM:ibp} and  indicates what kind of integration by parts will be used for reduction.
\begin{lemma} \label{LM:fibp}
Let $\xi$ be a monomial in $t_{1},$ \ldots, $t_n$ and $f \in K_0$. Then the followings hold.
\begin{itemize}
\item[(i)] $f^\prime \xi \equiv 0 \mod K_n^\prime + R_n^{(\xi)}.$
\item[(ii)]
Let $s = \scale_n(\xi)$. Then, for all $c_1, \ldots, c_s \in C$,
\[ (c_1 t_1^\prime + \cdots  + c_s t_s^\prime) \xi \equiv 0 \mod K_n^\prime + R_n^{(\xi)}. \]
\end{itemize}
\end{lemma}
\begin{proof}
(i)  It follows from integration by parts and the fact that ~$\xi^\prime$ belongs to~$R_n^{(\xi)}$.

(ii) Set $L_0=0$ and $L_i = \sum_{j=1}^i c_j t_j$  for $i=1$, \ldots, $n$.

If $\xi=1$, then $s=n$ and $L_n^\prime \xi \in K_n^\prime.$
The assertion clearly holds. Assume that $\xi=t_s^{e_s} \cdots t_n^{e_n}$ with $e_s>0$.
Then
$L_s^\prime \xi = L_{s-1}^\prime \xi + c_s t_s^\prime \xi.$
Note that $L_{s-1}^\prime \xi$ belongs to $K_n^\prime + R_n^{(\xi)}$ by a direct use of integration by parts.
Set $ \eta =  \xi/t_s^{e_s}$. Then the term~$c_s t_s^\prime \xi$
is equal to $\frac{c_s}{e_s+1} \left( t_s^{e_s+1} \right)^\prime \eta$.
Integration by parts leads to
\begin{equation} \label{EQ:decrease}
c_s t_s^\prime \xi  \equiv \frac{-c_s}{e_s+1} t_s^{e_s+1} \eta^\prime
 \mod K_n^\prime.
\end{equation}
Then $\eta=1$ if $e_j = 0$ for all $j$ with $j>s$. So $c_s t_s^\prime \xi$ belongs to $K_n^\prime$ by~\eqref{EQ:decrease}.
Otherwise, $e_j>0$ for some $j$ with $s < j \le n$.
Then each monomial in $t_s^{e_s+1} \eta^\prime$ is of total degree
$\sum_{j=s}^n e_j$ and is of degree $e_s+1$ in $t_s$. {  So~$t_s^{e_s+1} \eta^\prime \prec \xi$.
Consequently, $c_s t_s^\prime \xi \in K_n^\prime + R_n^{(\xi)}$ by~\eqref{EQ:decrease}.}
\end{proof}

In the rest of this section, we let $K_0=C(t_0)$ with the usual derivation in $t_0$. By Lemma~\ref{LM:rational},
we may further assume that $t_i^\prime$ is nonzero and $t_0$-simple for all $i$ with $1 \le i \le n$.
\begin{defn} \label{DEF:frigid}
For every $k$ with $1 \le k \le n$,
an element of~$K_0$ is said to be {\em $k$-rigid} if either it is equal to zero or
it is $t_0$-simple and is {  not a $C$-linearly combination of
$t_1^\prime, \ldots, t_k^\prime$.}
\end{defn}

\begin{prop} \label{PROP:fin-field0}
For $p \in R_n$, there exists $q \in R_n$  such that
$$p \equiv q \mod K_n^\prime$$
and that $\hc(q)$ is $s$-rigid, where $s = \scale_n(q)$. Moreover, $q \preceq p$.
\end{prop}
\begin{proof}
Set $q=0$ if $p=0$. Assume $p \neq 0$
and $\xi=\hm(p)$.
By the rational reduction, $\hc(p)=f^\prime + g$ for some $f, g \in K_0$ with $g$ being $t_0$-simple.
Then $p = f^\prime \xi + g \xi \mod R_n^{(\xi)}$. By Lemma~\ref{LM:fibp} (i), $p \equiv g \xi +r \mod K_n^\prime$
for some $r \in  R_n^{(\xi)}$. Set $s = \scale_n(\xi)$.
If $g$ is nonzero and $s$-rigid, then set $q=g \xi+r$.
Otherwise, $p \equiv \tilde r \mod K_n^\prime$ for some $\tilde r \in R_n^{(\xi)}$ by Lemma~\ref{LM:fibp} (ii). The proposition
follows from a direct Noetherian induction on $\hm(\tilde r)$ with respect to $\prec$.
\end{proof}

\begin{ex}\label{EX:flatpoly}
Let $K_0 {=} \bC(x)$, $t_1{=}\log(x)$, $t_2{=}\log(x+1)$.
and
$$p=t_1^2 t_2 + (2/x)t_1t_2 + ((2/(x+1))t_1.$$
Then $\hc(p) =1,$  which is not $1$-rigid.
Since  $t_1^2t_2 = x^\prime t_1 t_2^2$, integration by parts leads to $p = \left(xt_1^2t_2\right)^\prime + q$, where
$q = \left(\frac{2}{x}- 2\right)t_1t_2 - \frac{x}{x+1}t_1^2 + \frac{2}{x+1}t_1.$
We can then reduce $q$ further, because $\hc(q) = (2t_1-2x)^\prime$,  which is not $1$-rigid either.
Repeating this reduction a finite number of times, we see that $p$ is equal to the derivative of
$(x+1)t_1^2t_2-2xt_1t_2-xt_1^2+(2x+2)t_2+4xt_1-6x.$
\end{ex}
\subsection{Reduction}
A flat analogue of straight polynomials is given below.
\begin{defn} \label{DEF:fpoly}
A polynomial in $C[t_0]$ is said to be $t_0$-flat if it is zero. For $n \ge 1$,
$p \in K_{n-1}[t_n]$ is called a {\em $t_n$-flat polynomial}
if there exist $p_i \in K_i[t_{i+1}, \ldots, t_n]$ for all $i$ with $0 \le i \le n-1$
such that {  $p = \sum_{i=0}^{n-1} p_i$,
$\hc(p_i)$ is $t_i$-simple for all $i \ge 1$, and $\hc(p_0)$ is $s$-rigid,} where $s = \scale_n(p_0)$.
The sequence $\{p_i\}_{i=0, 1, \ldots, n-1}$ is called  a  {\em sequence associated to $p$.}
\end{defn}

\begin{ex} \label{EX:flat1}
Let $n=3$ and $t_0=x$, $t_1=\log(x)$, $t_2=\log(x+1)$ and $t_3=\log(x+2)$. Consider $p \in K_2[t_3]$
\[ p = \underbrace{\frac{1}{t_2} t_3^2}_{p_2} + \underbrace{\frac{1}{t_1} t_2 t_3}_{p_1} + \underbrace{\frac{1}{x+k} t_3^3 + x t_2 t_3}_{p_0}, \]
where $k \in \bZ$. Obviously, $\hc(p_2)$ is $t_2$-simple and $\hc(p_1)$ is $t_1$-simple.
Moreover, $\scale_3(p_0)=3$ and $\hc(p_0)$ is $3$-rigid
if $k \notin \{0,1,2\}$. So $p$ is $t_3$-flat if  $k \notin \{0,1,2\}$.
\end{ex}
We are going to extend the results in Section~\ref{SECT:str} to the flat case, based
on the following technical lemma.
\begin{lemma} \label{LM:rec}
Let $n \ge 1$ and $p \in K_{n-1}[t_n]$ be $t_n$-flat. Set $\ell$ to be $\lc_{t_n}(p)$.
Then $\fp_{t_{n-1}}(\ell)$ is $t_{n-1}$-simple and $\pp_{t_{n-1}}(\ell)$ is $t_{n-1}$-flat.
Moreover, $\pp_{t_{n-1}}(\ell) - ct_n^\prime$ is $t_{n-1}$-flat for all~$c \in C$ if $n >1$.
\end{lemma}
\begin{proof}
The lemma is trivial if $p =0$.

Assume that $p$ is nonzero and $d = \deg_{t_n}(p)$.
Let $\{p_i\}_{i=0, 1, \ldots, n-1}$ be a sequence associated to $p$,
and let $\ell_i$ be the coefficient of $t_n^d$ in $p_i$.
Evidently, $\ell = \sum_{i=0}^{n-1} \ell_i$, $\fp_{t_{n-1}}(\ell) = \ell_{n-1}$ and $\pp_{t_{n-1}}(\ell) = \sum_{i=0}^{n-2} \ell_i$.
Moreover,  $\ell_i = 0$ if $\deg_{t_n}(p_i) < d$, and $\hc(\ell_i) = \hc(p_i)$ otherwise.
This is  because $\prec$ is purely lexicographic with $t_{i+1} \prec \cdots \prec t_n$ for all $i$ with $0 \le i \le n-1$.
Thus, $\fp_{t_{n-1}}(\ell)$ is $t_{n-1}$-simple and $\pp_{t_{n-1}}(\ell)$ is $t_{n-1}$-flat by Definition~\ref{DEF:fpoly}.

It remains to show the second assertion. Assume  $n > 1$. Then
\begin{equation} \label{EQ:last}
\pp_{t_{n-1}}(\ell) - ct_n^\prime = \ell_{n-2}+\cdots+\ell_1+ \tilde{\ell}_0 \quad \text{with $\tilde{\ell}_0 =\ell_0-ct_n^\prime$},
\end{equation}
and  $\hc(\ell_j)$ is $t_j$-simple for all $j$ with $1 \le j \le n-2$.

Set $s = \scale_n(p_0)$ and $\tilde s = \scale_{n-1}(\tilde{\ell}_0)$.
It suffices to prove that $\hc(\tilde{\ell}_0)$ is $\tilde s$-rigid by~\eqref{EQ:last} and  Definition~\ref{DEF:fpoly}.


\smallskip  \noindent
{\em Case 1.} $\ell_0 \notin K_0$.   Then $s < n$,
$$\hm(p_0) = t_s^{e_s} \cdots t_{n-1}^{e_{n-1}} t_{n}^{d} \quad \text{and} \quad \hm(\ell_0) = t_s^{e_s}  \cdots t_{n-1}^{e_{n-1}},$$
where $e_s >0$.  Moreover,
$$s = \scale_{n-1}(\ell_0), \,\, \hm(\ell_0) = \hm(\tilde{\ell}_0), \,\, \text{and} \,\,
\hc(p_0)=\hc(\ell_0)=\hc(\tilde{\ell}_0).$$
In particular, $\tilde{s} = s$. Hence, $\hc(\tilde{\ell}_0)$ is  $\tilde s$-rigid, because $\hc(p_0)$ is $s$-rigid.

\smallskip \noindent
{\em Case 2.}  $\ell_0 \in K_0$ with $\ell_0 \neq 0$.
Then $\hm(p_0) = t_n^d$ and $s=n$.
Moreover, $\tilde{s}=n-1$, since $\tilde{\ell}_0 \in K_0$.
Note that $p$ is $t_n$-flat. So  $\hc(p_0)$ is not a $C$-linear combination of $\{t_1^\prime, \ldots, t_{n-1}^\prime, t_n^\prime\}$,
and neither is $\ell_0$ because $\ell_0 = \hc(p_0)$.
Consequently,  $\tilde{\ell}_0$ is not a $C$-linear combination of $\{t_1^\prime, \ldots, t_{n-1}^\prime\}$, and neither
is $\hc(\tilde{\ell}_0)$, because $\hc(\tilde{\ell}_0) = \tilde{\ell}_0$. Thus, $\hc(\tilde{\ell}_0)$ is  $(n-1)$-rigid.

\smallskip \noindent
{\em Case 3.} $\ell_0 = 0$. Then $\tilde s = n-1$ and $\hc(\tilde{\ell}_0) = \tilde{\ell}_0 = -c t_n^\prime$, which is $(n-1)$-rigid by Remark~\ref{RE:indep}.

The second assertion is proved.
\end{proof}
The next lemma is a flat-analogue of Lemma~\ref{LM:rigid}
\begin{lemma} \label{LM:flatlc}
Let $n \ge 1$ and $p \in K_{n-1}[t_n]$ be $t_n$-flat. If
\begin{equation} \label{EQ:flatlc}
\lc_{t_n}(p) \equiv c t_n^\prime \mod K_{n-1}^\prime
\end{equation}
 for some $c \in C$,
then both $p$ and $c$ are zero.
\end{lemma}
\begin{proof}
If $n=1$, then the tower $K_0 \subset K_1$ is also straight, and
$p$ is $t_1$-straight by Definition~\ref{DEF:spoly} and Lemma~\ref{LM:rational}.
Both~$p$ and $c$ are zero by Lemma~\ref{LM:rigid}.

Assume that $n>1$ and the lemma holds for $n-1$.
Set $\ell = \lc_{t_n}(p)$.
Applying the map $\hp_{t_{n-1}}$ to~\eqref{EQ:flatlc}, we have $\hp_{t_{n-1}}(\ell)=0$.
Then  $\fp_{t_{n-1}}(\ell)= 0$ by the first assertion of Lemma~\ref{LM:rec}. Consequently,  we have $\ell = \pp_{t_{n-1}}(\ell)$.
Let $q = \ell - ct_n^\prime$. Then $q$ is $t_{n-1}$-flat by the second assertion of Lemma~\ref{LM:rec}.
On the other hand, $q \in K_{n-1}^\prime$ by~\eqref{EQ:flatlc}. Then $\lc_{t_{n-1}}(q) \equiv \tilde c t_{n-1}^\prime \mod K_{n-2}^\prime$ for some $\tilde c \in C$ by Lemma~\ref{LM:lc}.
So $q=0$ by the induction hypothesis.
Accordingly,
\begin{equation} \label{EQ:lcf}
\ell = ct_n^\prime \in K_0.
\end{equation}
Let $\{p_i\}_{i=0,1,\ldots,n-1}$ be a sequence associated to $p$.
By~\eqref{EQ:lcf},  we have $\hc(p_0)=c t_n^\prime$, because $c t_n^\prime$ is not $t_i$-simple for all $i$ with $1 \le i \le n$.
Hence, $\hm(p_0)$ is a nonnegative power of $t_n$ and $\scale_n(p_0)=n$.
Then $c t_n^\prime$ is $n$-rigid by Definition~\ref{DEF:fpoly}. We have $c=0$.  By~\eqref{EQ:lcf},  we conclude  that $\ell$ is zero, and so is $p$.
\end{proof}

The following proposition corresponds to Proposition~\ref{PROP:spoly}.
\begin{prop} \label{PROP:fsin-field}
{Let $n \ge 1$} and $p$  be a $t_n$-flat polynomial in~$K_{n-1}[t_n]$. If $p \in K_n^\prime$, then $p=0$.
\end{prop}
\begin{proof}
Since $p \in K_n^\prime$,
we have $\lc_{t_n}(p) \equiv c t_n^\prime \mod K_{n-1}^\prime $
for some $c \in C$
by Lemma~\ref{LM:lc}. Then $p=0$ by Lemma~\ref{LM:flatlc}.
\end{proof}
The next lemma corresponds to Lemma~\ref{LM:spoly}.
\begin{lemma} \label{LM:fspoly}
For $p \in K_{n-1}[t_n]$, there exists a $t_n$-flat polynomial $q \in K_{n-1}[t_n]$ such that
$p \equiv q \mod K_n^\prime.$
Moreover, $\deg_{t_n}(q)$ is no more than $\deg_{t_n}(p)$.
\end{lemma}
\begin{proof}
By Lemma~\ref{LM:seq}, there exist $p_i \in K_i[t_{i+1}, \ldots, t_n]$ for all $i$ with $1 \le i \le n-1$ and $p_0 \in R_{n}$ such that
$$p \equiv \sum_{i=1}^{n-1} p_i + p_0\mod K_n^\prime.$$
Moreover, $\hc(p_i) \in K_i$ is $t_i$-simple for all $i \ge 1$.
By Proposition~\ref{PROP:fin-field0}, there exists $r \in R_n$ with $s = \scale_n(r)$
such that $p_0 \equiv r \mod K_n^\prime$ and that $\hc(r)$ is $s$-rigid.
Set $q$ to be $\sum_{i=1}^{n-1} p_i + r $. Then $q$ is $t_n$-flat and $p \equiv q \mod K_n^\prime.$
\end{proof}
\begin{ex} \label{EX:flat2}
Let $p$ be given as in Example~\ref{EX:flat1}, in which $k=2$.
By integration by parts, we have
\[ p \equiv  p_2 + p_1 + \underbrace{ - 3 t_3^\prime t_2 t_3^2 + x t_2 t_3}_{q_0} \mod K_3^\prime. \]
Then $\scale_3(q_0)=2$ and $\hc(q_0) = - 3 t_3^\prime = -3/(x+2)$, which is $2$-rigid.
Hence, $p_2 + p_1 + q_0$ is $t_3$-flat.
\end{ex}

We are ready to present the main result of this section.
\begin{thm} \label{TH:fadd}
For $f \in K_n$, the following assertions hold.
\begin{itemize}
\item[(i)]
There exist a $t_n$-simple element $g \in K_n$ and a $t_n$-flat polynomial $p \in K_{n-1}[t_n]$ such that
\begin{equation} \label{EQ:fadd}
f \equiv g + p \mod K_n^\prime.
\end{equation}
\item[(ii)] $f \equiv 0 \mod K_n^\prime$ if and only if both $g$ and $p$ are zero.
\item[(iii)] If  $f \equiv \tilde g + \tilde p \mod K_n^\prime$, where $\tilde g \in K_n$ is $t_n$-simple
 and $\tilde p \in K_{n-1}[t_n]$, then $g = \tilde g$ and $\deg_{t_n}(p) \le \deg_{t_n}(\tilde p).$
\end{itemize}
\end{thm}
\begin{proof} (i) Applying Algorithm {\tt HermiteReduce} to $f$ with respect to~$t_n$, we get a $t_n$-simple element $g$ of~$K_n$
and an element $h$ of~$K_{n-1}[t_n]$ such that $f \equiv g + h \mod K_n^\prime.$
We can replace $h$ with a $t_n$-flat polynomial~$p$
 by Lemma~\ref{LM:fspoly}.

(ii) Assume $f \in  K_n^\prime$. Then \eqref{EQ:fadd} becomes $g + p \equiv 0 \mod K_n^\prime$. Applying the map $\hp_{t_n}$ to the above congruence
yields $g = 0$ by Lemma~\ref{LM:hp}. Thus, $p \equiv 0 \mod K_n^\prime$. Consequently, $p=0$ by Proposition~\ref{PROP:fsin-field}.

(iii) Since $(g - \tilde{g}) + (p - \tilde{p}) \equiv 0 \mod K_n^\prime$ and $g  - \tilde g$ is $t_n$-simple, we have $g = \tilde g$ by Lemma~\ref{LM:hp}.
So $p - \tilde p \equiv 0 \mod K_n^\prime$.  By Lemma~\ref{LM:lc}, we have  $\lc_{t_n}(p - \tilde p) \equiv ct_n^\prime \mod K_{n-1}^\prime$ for some $c \in C$.
If $\deg_{t_n}(\tilde p)$ is smaller than~$\deg_{t_n}(p)$, then $\lc_{t_n}(p) = \lc_{t_n}(p - \tilde p) \equiv c t_n^\prime \mod K_{n-1}^\prime$.
By Lemma~\ref{LM:flatlc}, we conclude $p=0$, a contradiction.
\end{proof}
\section{Elementary integrability} \label{SECT:elem}
In this section, we study elementary integrability of elements in a straight or flat tower
by Theorems~\ref{TH:sadd} and~\ref{TH:fadd}.

\begin{thm} \label{TH:elem}
{Let the tower be given in~\eqref{EQ:ext}, in which $C$ is algebraically closed, $K_0=C(t_0)$ and $t_i$ be a $C$-linear combination of logarithmic monomials over $K_{i-1}$ for all $i$ with $1 \le i \le n$.}
Assume that, for $f \in K_n$,
\begin{equation} \label{EQ:elem}
f \equiv g + p \mod K_n^\prime,
\end{equation}
where $g$ and $p$ are described in~\eqref{EQ:sadd} if \eqref{EQ:ext} is straight, and in~\eqref{EQ:fadd} if \eqref{EQ:ext} is flat, respectively.
Then $f$ is elementarily integrable over $K_n$ if and only if $g+p$ is a $C$-linear combination of logarithmic derivatives in $K_n$.
\end{thm}
\begin{proof}
We denote by $\bL_i$ the $C$-linear subspace spanned by the logarithmic derivatives in $K_i$
for all $i$ with $0 \le i \le n$.

Clearly, $f$ is elementarily integrable over $K_n$ if  $g+p \in \bL_n$.

Conversely,
there exists $r \in \bL_n$ such that $ f \equiv r \mod K_n^\prime$ by Liouville's theorem {  (\cite[Theorem 5.5.1]{BronsteinBook})}.
{  By~\eqref{EQ:elem},
\begin{equation} \label{EQ:logder1}
g + p \equiv r \mod K_n^\prime.
\end{equation}
Note that $\hp_{t_n}(g+p) = g$, as $g$ is $t_n$-simple.
So $g = \hp_{t_n}(r)$ by~\eqref{EQ:logder1} and Lemma~\ref{LM:hp}. Hence, $g \in \bL_n$ by Lemma~\ref{LM:ld}.}
Set $\tilde r$ to be  $r - \hp_{t_n}(r)$.
Then $\tilde r \in \bL_{n-1}$ by Lemma~\ref{LM:ld}.
Moreover, \eqref{EQ:logder1} becomes
\begin{equation} \label{EQ:logder2}
p \equiv \tilde r \mod K_n^\prime.
\end{equation}
We show that~\eqref{EQ:logder1} implies  $g+p \in \bL_n$ by induction.
If $n=0$, then~$p$ is zero. The assertion holds.
Assume that the assertion holds for~$n-1$.
Let $d = \deg_{t_n}(p)$ and $\ell = \lc_{t_n}(p)$.

\smallskip \noindent
{\em Case 1}.\ $d>0$. Then $\ell = \lc_{t_n}(p-\tilde r)$, which,
together with~\eqref{EQ:logder2} and Lemma~\ref{LM:lc}, implies that $\ell \equiv c t_n^\prime \mod K_{n-1}^\prime$ for some $c \in C$.
Then $\ell$ is equal to $0$ by Lemma~\ref{LM:rigid} in the straight case and by Lemma~\ref{LM:flatlc} in the flat case, a contradiction.

\smallskip \noindent
{\em Case 2}.\ $d=0$. Then $\ell=p$.
So
$\ell \equiv \tilde r + ct_n^\prime \mod K_{n-1}^\prime$ for some $c \in C$ by~\eqref{EQ:logder2} and  Lemma~\ref{LM:lc}.
Consequently, $\ell$ is elementarily integrable over $K_{n-1}$, because $\tilde r, t_n^\prime \in \bL_{n-1}$.
Moreover, the above congruence can be rewritten as
\begin{equation} \label{EQ:logder3}
 \fp_{t_{n-1}}(\ell) + \pp_{t_{n-1}}(\ell) \equiv \tilde r + ct_n^\prime \mod K_{n-1}^\prime
 \end{equation}
Note that $\fp_{t_{n-1}}(\ell)$ is $t_{n-1}$-simple and $\pp_{t_{n-1}}(\ell)$ is $t_{n-1}$-straight (resp.\ flat)
by Definition~\ref{DEF:spoly} (resp.\ Lemma~\ref{LM:rec}).
By~\eqref{EQ:logder3} and the induction hypothesis, we see that $\ell \in \bL_{n-1}$. Therefore, $p$ belongs to $\bL_{n-1}$.
Accordingly, $g+p \in \bL_n$.
\end{proof}

Determining whether an element $r$ in $K_n$ is a $C$-linear combination of
logarithmic derivatives amounts to computing partial fraction decompositions and
Rothstein--Trager resultants by Theorem 4.4.3 in~\cite{BronsteinBook}.
So elementary integrability in straight and flat towers can be checked by merely
algebraic computation whenever a decomposition in the form \eqref{EQ:elem} is available.
\begin{ex} \label{EX:elem2}
Let $K_0$, $t_1$ and $t_2$ be given as in Example~\ref{EX:flatpoly}.
We compute an additive decomposition for
$$f =  \frac{1}{x t_1}+ \frac{1}{x t_2 +t_2} + t_1^2 t_2 + \frac{2}{x}t_1 t_2 + \frac{2}{x+1}t_1 + \frac{1}{x+2}.$$
By Theorem~\ref{TH:fadd} and Example~\ref{EX:flatpoly}, we have
\[  f = a^\prime + \underbrace{  \frac{1}{x t_2 +t_2} }_g+\underbrace{ \frac{1}{x t_1}+\frac{1}{x+2}}_p, \]
where $a =(x+1)t_1^2t_2-2xt_1t_2-xt_1^2+(2x+2)t_2+4xt_1-6x$. As the Rothstein--Trager resultant of each fraction in $g+p$
has only constant roots, $g+p$ is a $C$-linear combination of logarithmic derivatives
in~$K_2$. So $f$ is elementarily integrable over $K_2$ by Theorem~\ref{TH:elem}. Indeed,
\[ \int f \, dx = a + \log(t_2 ) + \log(t_1) + \log(x+2).\]
\end{ex}
\section{Telescopers for  elementary  functions}\label{SECT:ct}

The problem of creative telescoping is classically formulated for $D$-finite functions in terms of linear differential operators~\cite{Almkvist1990, Zeilberger1991}.
Raab in his thesis~\cite{Raab2012} has studied the telescoping problem viewed as a special case of the parametric integration problem in
differential fields. However, there are no theoretical results concerning the existence of telescopers for elementary functions.
To be more precise, let $F$ be a differential field with two derivations $D_x$ and~$D_y$ that commute with each other
and let $F_\partial$ be the set $\{f\in F \mid \partial(f) = 0\}$ for $\partial \in \{D_x, D_y\}$.
For a given element $f\in F$, the telescoping problem asks whether there exists a nonzero linear differential
operator $L = \sum_{i=0}^d \ell_i D_x^i$ with $\ell_i\in F_{D_y}$ such that $L(f)=D_y(g)$ for some~$g$ in a  specific differential extension~$E$ of~$F$.
We call $L$ a \emph{telescoper} for~$f$ and $g$ the corresponding \emph{certificate} for $L$ in~$E$.
Usually, we take~$E$ to be the field $F$ itself or an elementary extension of $F$.
In contrast to $D$-finite functions,  telescopers may not exist for elementary functions as shown in the following example.

\begin{ex} Let $F = \bC(x, y)$ and $E = F(t_1, t_2)$ be a differential field extension of $F$ with
$$t_1 = \log(x^2+y^2) \quad \text{and} \quad  t_2 = \log(1 + t_1).$$
We first show that $f= 1/t_1 \in F(t_1)$ has no telescoper with certificate in any elementary extension of $F(t_1)$.
Since $t_1$ is a primitive monomial over $F$, we have $F_{D_y} = \bC(x)$.
We claim that for any $i\in \bN$, $D_x^i(f)$ can be decomposed as
\[D_x^i(f) = D_y(g_i) + \frac{a_i}{t_1},\]
where~$g_i\in F(t_1)$, and $a_i \in F$ satisfies the recurrence relation
\[a_{i+1} = D_x(a_i)-D_y\left(\frac{xa_i}{y}\right)\quad \text{with~$a_0 = 1$.}\]
For $n=0$, the claim holds by taking $g_0 = 0$.
Assume that the claim holds for all $i<k$. Applying the induction hypothesis and
Algorithm {\tt HermiteReduce} to $D_x^{k}(f)$ yields
\begin{align*}
 D_x^{k}(f)  & = D_x(D_x^{k-1}(f)) = D_x\left(D_y(g_{k-1})+ \frac{a_{k-1}}{t_1}\right) \\
   & {=} D_y\left(D_x(g_{k-1}) + \frac{a_{k-1}x}{yt_1}\right){+} \frac{D_x(a_{k-1}){-}D_y(\frac{x a_{k-1}}{y})}{t_1}.
\end{align*}
This completes the induction. A straightforward calculation
shows that $a_i = A_i/y^{2i}$ for some $A_i\in \bC[x, y]\setminus \{0\}$ with $\deg_y(A_i)<2i$.
Using the notion of residues in~\cite[page 118]{BronsteinBook}, we have
 \[ \operatorname{residue}_{t_1}\left(\frac{a_i}{t_1}\right) = \frac{a_i}{D_y(t_1)} = \frac{(x^2+y^2)A_i}{2y^{2i+1}},\]
which is not in $\bC(x)$. Then $D_x^i(f)$ is not elementarily integrable over $F(t_1)$ for any $i\in \bN$
by the residue criterion in~\cite[Theorem 5.6.1]{BronsteinBook}. Assume that $f$ has a telescoper $L := \sum_{i=0}^d \ell_i D_x^i$
with $\ell_i \in \bC(x)$ not all zero. Then $L(f)$ is elementarily integrable over $F(t_1)$. However,
\[L(f) =D_y\left(\sum_{i=0}^d \ell_i g_i\right) + \frac{\sum_{i=0}^d \ell_i a_i}{t_1}.\]
Since all of the $\ell_i$'s are in $\bC(x)$ and $\gcd(x^2+y^2, y^m)=1$ for any $m\in \bN$, the residue of $\sum_{i=0}^d \ell_i a_i/t_1$ is not  in $\bC(x)$, which implies
that~$L(f)$ is not elementarily integrable over~$F(t_1)$, a contradiction.

We now show that $p = f t_2 + 1\in F(t_1)[t_2]$ has no telescoper with certificate in any elementary extension of $F(t_1, t_2)$.
Since $t_2$ is also a primitive monomial over $F(t_1)$, we have $E_{D_y} = \bC(x)$.
Assume that $L := \sum_{i=0}^d \ell_i D_x^i$ with $\ell_i \in \bC(x)$ not all zero
is a telescoper for $p$. Then $L(p)$ is elementarily integrable over~$E$.
By a direct calculation, we get
$L(p) = L(f)t_2 + r$ with $r\in F(t_1)$. The elementary integrability of $L(p)$ implies that
$L(f) = cD_y(t_2) + D_y(b)$ for some $c\in \bC(x)$ and $b\in F(t_1)$ by the formula (5.13)
in the proof of Theorem 5.8.1 in~\cite[page 157]{BronsteinBook}. We claim that~$c=0$.
Since $D_x^i(f)= u_{i}/t_1^{i+1}$ with $u_i\in F[t_1]$ and $\deg_{t_1}(u_i)< i+1$ and $D_y(t_2) = D_y(t_1)/(1+t_1)$,
the orders of $D_x^i(f)$ and $D_y(t_2)$ at $1+t_1$ are equal to $0$ and $1$, respectively.
If $c$ is not zero, the order of $cD_y(t_2)$ at $1+t_1$ is equal to $1$, which does not match with
that of $L(f)-D_y(b)$ by Theorem 4.4.2~(i) in \cite{BronsteinBook}, a contradiction. Then $L(f)=D_y(b)$, i.e.,
$L$ is a telescoper for~$f$. This contradicts with the first assertion we have shown.
\end{ex}
The next example shows that additive decompositions in Theorems~\ref{TH:sadd} and~\ref{TH:fadd}
are useful for detecting the existence of telescopers for elementary functions that are
not $D$-finite.

\begin{ex} Let~$F = \bC(x, y)$ and $E = F(t)$ be a differential field extension of $F$
with~$t = \log(x^2+y^2)$. Consider the function
$f = t+1 - \frac{2y}{(x^2+y^2)t^2}$.
Since the derivatives $D_x^i(1/t^2) = a_i/t^{i+2}$ with $a_i\in F\setminus\{0\}$  are linearly independent over $F$,
we see that $1/t^2$ is not $D$-finite over $F$, and nether is $f$.
Note that $f$ can be decomposed as
\[f = D_y(1/t) + t+1.\]
Since $t+1$ is D-finite, it has a telescoper, and so does $f$.
\end{ex}

\section{Conclusion}
In this paper, we developed additive decompositions in straight and flat towers, which enable us to
determine in-field and elementary integrability in a straightforward manner.
It is natural to ask whether one can develop an additive decomposition in a general primitive tower. Moreover, we plan to investigate about the existence and the construction of telescopers for elementary functions using additive decompositions.

\bibliographystyle{plain}



{

\def\cprime{$'$}

}

\end{document}